\documentclass[aps,amssymb,amsmath,amsfonts,twocolumn,pra,showpacs]{revtex4}

\usepackage{tikz}
\usetikzlibrary{calc,fadings,decorations.pathreplacing,shapes.geometric}

\newcommand\pgfmathsinandcos[3]{%
  \pgfmathsetmacro#1{sin(#3)}%
  \pgfmathsetmacro#2{cos(#3)}%
}
\newcommand\LongitudePlane[3][current plane]{%
  \pgfmathsinandcos\sinEl\cosEl{#2} 
  \pgfmathsinandcos\sint\cost{#3} 
  \tikzset{#1/.estyle={cm={\cost,\sint*\sinEl,0,\cosEl,(0,0)}}}
}
\newcommand\LatitudePlane[3][current plane]{%
  \pgfmathsinandcos\sinEl\cosEl{#2} 
  \pgfmathsinandcos\sint\cost{#3} 
  \pgfmathsetmacro\yshift{\cosEl*\sint}
  \tikzset{#1/.estyle={cm={\cost,0,0,\cost*\sinEl,(0,\yshift)}}} %
}
\newcommand\DrawLongitudeCircle[2][1]{
  \LongitudePlane{\angEl}{#2}
  \tikzset{current plane/.prefix style={scale=#1}}
  \pgfmathsetmacro\angVis{atan(sin(#2)*cos(\angEl)/sin(\angEl))} %
  \draw[current plane] (\angVis:1) arc (\angVis:\angVis+180:1);
  \draw[current plane,dashed] (\angVis-180:1) arc (\angVis-180:\angVis:1);
}
\newcommand\DrawLatitudeCircle[2][1]{
  \LatitudePlane{\angEl}{#2}
  \tikzset{current plane/.prefix style={scale=#1}}
  \pgfmathsetmacro\sinVis{sin(#2)/cos(#2)*sin(\angEl)/cos(\angEl)}
  \pgfmathsetmacro\angVis{asin(min(1,max(\sinVis,-1)))}
  \draw[current plane] (\angVis:1) arc (\angVis:-\angVis-180:1);
  \draw[current plane,dashed] (180-\angVis:1) arc (180-\angVis:\angVis:1);
}

\tikzset{%
  >=latex, 
  inner sep=0pt,%
  outer sep=2pt,%
  mark coordinate/.style={inner sep=0pt,outer sep=0pt,minimum size=3pt,
    fill=black,circle}%
}

\def\R{1.5} 
\def\angEl{35} 
\def\angAz{-105} 
\def\angPhi{-40} 

\def\ellipsoidInBall{
\def\R{2.5}
\begin{tikzpicture} 

\begin{scope}[transform canvas={rotate=-20}]
  \shade[shading=ball, ball color = gray!15] (0,0) ellipse (0.5*2.5 and 0.3*2.5);

\draw (0,0) ellipse (0.5*2.5 and 0.1*2.5);
\draw (0,0) ellipse (0.1*2.5 and 0.3*2.5);

\draw[->] (-0.5*\R,0) -- (0.5*\R,0) node[above] {$\eta_1$};
\draw[->] (0,-0.3*\R) -- (0,0.3*\R) node[above] {$\eta_2$};
\draw[->] (-0.1*\R,-0.1*\R) -- (0.1*\R,0.1*\R) node[above] {$\eta_3$};

\end{scope}

\pgfmathsetmacro\H{\R*cos(\angEl)} 
\tikzset{xyplane/.estyle={cm={cos(\angAz),sin(\angAz)*sin(\angEl),-sin(\angAz),
                              cos(\angAz)*sin(\angEl),(0,-\H)}}}
\LongitudePlane[xzplane]{\angEl}{\angAz}
\LongitudePlane[pzplane]{\angEl}{\angPhi}
\LatitudePlane[equator]{\angEl}{0}
\draw (0,0) circle (\R);
\DrawLatitudeCircle[\R]{0} 
\DrawLongitudeCircle[\R]{\angAz} 
\DrawLongitudeCircle[\R]{\angAz+90} 

\end{tikzpicture}
}

\usepackage{graphicx}
\usepackage{bm,bbm}
\usepackage{epstopdf}
\usepackage{amsthm}
\usepackage{bbold}
\usepackage{amsmath}

\usepackage{hyperref}
\usepackage{pst-all}
\usepackage{graphicx}
\usepackage{graphics}
\usepackage{epsfig}
\usepackage{color}
\usepackage{epsf}
\newcommand{\1}{\mathbbm{1}} 

\usepackage[T1]{fontenc}
\usepackage[english]{babel}
\usepackage[latin2]{inputenc}

\newtheorem{proposition}{Proposition}
\newtheorem{lemma}{Lemma}

\newtheorem{theorem}{Theorem}

\newtheorem{corollary}[proposition]{Corollary}

\renewcommand{\c}[1]{\mathcal{#1}}

\newcommand{\idty}{\1}

\DeclareMathOperator*{\map}{map}
\DeclareMathOperator*{\rec}{rec}
\newcommand{\<}{\langle}
\renewcommand{\>}{\rangle}

\newcommand{\scalar}[2]{\langle #1 | #2 \rangle}
\newcommand{\ketbra}[2]{| #1 \rangle \langle #2 |}
\newcommand{\ket}[1]{| #1 \rangle}
\newcommand{\bra}[1]{\langle #1 |}
\newcommand{\tr}{\mathrm{Tr}}

\newcommand{\stackidx}[4]{
\substack{
#1 #2 \\
#3 #4}
}

\begin{document}
\title{Entropic trade--off relations for quantum operations}

\author{Wojciech Roga$^{1,2}$, Zbigniew~Pucha{\l}a$^3$, {\L}ukasz Rudnicki$^4$, Karol {\.Z}yczkowski$^{2,4}$}

\affiliation{
$^1$Universit\`a degli Studi di Salerno, Via Ponte don Melillo, I-84084 Fisciano (SA), Italy \\
$^2$Institute of Physics, Jagiellonian University, ul.\ Reymonta 4, 30-059 Krak\'ow, Poland\\
$^3$Institute of Theoretical and Applied Informatics, Polish Academy
of Sciences, Ba{\l}tycka 5, 44-100 Gliwice, Poland  \\
$^4$Center for Theoretical Physics, Polish Academy of Sciences,
al.\ Lotnik\'ow 32/46, 02-668 Warszawa, Poland}

\date{05--02--2013, ver. 19.2}

\begin{abstract}
Spectral properties of an arbitrary matrix can be characterized by the entropy
of its rescaled singular values. Any quantum operation can be described by the
associated dynamical matrix or by the corresponding superoperator. The entropy
of the dynamical matrix describes  the degree of decoherence introduced by the
map, while the entropy of the superoperator characterizes the a
priori knowledge of the receiver  of the outcome of a quantum channel $\Phi$.
We prove that for any map acting on  a $N$--dimensional quantum system the sum of
both entropies is not smaller than $\ln N$. For any bistochastic map this lower
bound reads $2\ln N$. We investigate also the corresponding R{\'e}nyi entropies,
providing an upper bound for their sum and analyze entanglement
of the bi-partite quantum state associated with the channel.
\end{abstract}

\pacs{03.67.Hk 02.10.Ud 03.65.Aa}
\maketitle

\section{Introduction}

From the early days  of quantum mechanics the uncertainty principle was one 
of its the most significant features, 
as it shows in what respect the quantum theory differs from its
classical counterpart.
It was manifested that
the variances of the two noncommuting observables cannot be simultaneously arbitrarily small.
Therefore, if we prepare the quantum state as an eigenstate of one observable
we get the perfect knowledge about the corresponding physical quantity, however, we loose 
ability to predict the effect of measurement of the second, noncommuting observable. 
Preparing a state one shall always consider some trade--off 
regarding the observables which will be specified in the experiment.
Limits for such a trade--off  have been formulated as different uncertainty relations 
 \cite{heisenberg,robertson,BBM75,De83,maassen}. However, we shall point out that
 these uncertainty relations are not necessarily related to noncommuting observables,
 but may describe the trade--off originating from different descriptions 
of the same quantum state.  
As an example consider the entropic uncertainty relation \cite{BBM75} derived
for two probability distributions related to the same quantum state in position and momentum representations. 

The original formulation of Heisenberg \cite{heisenberg} of the uncertainty principle
concerns the product of variances of two non-commuting observables.
Assume that a physical system is described by a quantum state $|\psi\>$, and
several copies of this state are available.  Measuring an observable $A$ in each
copy of this state results with the standard deviation $\Delta_{\psi} A
=\sqrt{\<\psi|A^2|\psi\>-\<\psi|A|\psi\>^2}$,  while $\Delta_{\psi} B$ denotes
an analogous expression for another observable $B$. According to the approach of
Robertson \cite{robertson} the product of these deviations is bounded from
below,
\begin{equation}
\Delta_{\psi}A\ \Delta_{\psi} B\geq\frac{1}{2}\left|\<\psi|[A,B]|\psi\>\right|,
\end{equation}
where $[A,B]=AB-BA$ denotes the commutator.  If the operators $A$ and $B$ do not
commute it is thus impossible to specify simultaneously precise values of both
observables.

Uncertainty relations can also be formulated for other quantities characterizing
the distributions of the measurement outcomes. One possible choice is to use
entropy which leads to entropic uncertainty relations of  Bia{\l}ynicki--Birula
and Mycielski \cite{BBM75}. This formulation can be considered as a
generalization of the standard approach as it implies the relations of
Heisenberg.

In the case of a finite dimensional Hilbert space the uncertainty relation can
be formulated in terms of the Shannon entropy.  Consider a non-degenerate
observable $A$, the eigenstates $|a_i\rangle$ of which determine an orthonormal
basis.  The probability that this observable measured in a pure state
$|\psi\rangle$ gives the$i$th outcome reads $a_i=|\<a_i|\psi\>|^2$. The
non-negative numbers $a_i$ satisfy $\sum_{i=1}^N a_i=1$,  so this distribution
can be characterized by the  Shannon entropy, $H(A)=-\sum_i a_i\ln a_i$. Let
$H(B)$ denotes the Shannon entropy  corresponding the the probability vector
$b_i=|\<b_i|\psi\>|^2$ associated with an observable $B$. If both observables do
not commute the sum of both entropies is bounded from below, as shown by Deutsch
\cite{De83}.  His result was improved by Maassen and Uffink \cite{maassen},  who
proved that
\begin{equation}
\label{masuuf}
H(A) + H(B) \ \geq \ -2\ln{c},
\end{equation}
where $c^2=\max_{j,k}|\<a_j|b_k\>|^2$ denotes the maximal overlap between the
eigenstates of both observables. Note that this bound depends solely on the
choice of the observables  and not on the state $|\psi\rangle$. Recent reviews
on entropic uncertainty relations can be found in \cite{WW10,BR11}, while a link
to stabilizer formalism was discussed in \cite{NKG12}. Certain generalizations
of uncertainty relations for more than two spaces can be based on the strong
subadditivity of entropy \cite{FL12}.

Relation (\ref{masuuf}) describes a bound for the information which can be
obtained in two non-complementary projective measurements.
Entropic uncertainty relations formulated for a pair of arbitrary measurements,
described by positive operator valued measures (POVM),
were obtained by Krishna and Parthasarathy \cite{KP02}.
A more general class of inequalities was derived later
by Rastegin \cite{Ra10,Ra11}.
Related recent results \cite{BCCRR10,TR11,CYGG11,CCYZ12}
concerned sum of two conditional entropies
characterizing two quantum measurements
described in terms of their POVM operators.

The so called {\it collapse of wave function} during the measurement 
is often considerd as another characteristic feature of quantum mechanics.
This postulate of 
quantum theory implies that the measurement disturbs the quantum state 
subjected to the process of quantum measurement.
In general, describing a quantum operation performed on an arbitrary state
one can consider a kind of trade--off relations between the efficiency of the measurement 
and the  disturbance introduced to the  measured states.

Even though the trade--off relations were investigated from the beginnings of quantum mechanics,
this field became a subject of a  considerable scientific interest in the recent 
decade \cite{Fuchs2001,DAriano2003,Buscemi2009}.   
The notion of {\sl disturbance}
of a state introduced by Maccone \cite{Maccone2006,Maccone2007},
can be related to the average fidelity
between an initial state of the system and and the state after the measurement \cite{Maccone2006,Buscemi2009}.
Another version of disturbance can be defined as a difference between the initial entropy of a quantum state
and the coherent information between the system and the measuring apparatus  \cite{Maccone2007}. 

In this work we will investigate
a single measurement process  described by a quantum operation:  a complete
positive, trace preserving linear map which acts on an input state of size $N$.
We attempt to compare the information loss introduced by the map 
(disturbance) with the
information the receiver knows about the outgoing state before the measurement
(the information gained by the apparatus).
The former quantity  can be characterized \cite{ZB04} by the entropy of a map
$S^{\map}(\Phi)$, equal to the von Neumann entropy of the quantum
state which corresponds  to the considered map by the Jamio\l{}kowski isomorphism
\cite{jamiolkowski,BZ06}. The latter quantity  will be described
by the \emph{singular quantum entropy} $S^{\rec}(\Phi)$
of Jumarie \cite{Ju00}, given by the
Shannon entropy of the normalized vector of singular values of the superoperator
matrix. We are going to show that the sum of these two  entropies is bounded
from below by $\ln N$.
Note that our approach concerns a given quantum map $\Phi$,
but it does not depend on the particular choice of the
Kraus operators (or POVM operators) used to represent the
quantum operation. We also derive an upper bound for the sum of entropies
$S^{\map}(\Phi)$ and $S^{\rec}(\Phi)$ and analyze entanglement properties
of the corresponding Jamio\l{}kowski-Choi state.

Our paper is organized as follows. In Section II we review basic concepts on
quantum maps, define entropies investigated and present a connection with
trade--off 
relations for quantum measurements. A motivation for our study
stems from investigations of the one--qubit maps presented in Section III.
General entropic inequalities for arbitrarily reordered matrices
are formulated in Section IV.
Main results of the work are contained in Section~V, in
which the  trade--off 
relations for quantum channels are derived
and entanglement of the corresponding states is analyzed.
Discussion of some other properties of the dynamical matrix
and a bound for the entropy of a map are
relegated to Appendices.

\section{Quantum operations and entropy}\label{secII}

A quantum state is described by a density matrix -- a~Hermitian, positive
semi-definite matrix of trace one. A~density matrix of dimension $N$ represents
the operator acting on $\c H_N$. The set of density matrices of dimension $N$ is
denoted as:
\begin{equation}
\c M_N=\{\rho: \rho=\rho^\dagger, \rho\geq0, \tr\rho=1\}.
\end{equation}
A \emph{quantum operation}  $\Phi$, also called a \emph{quantum channel}, is
defined as a completely positive (CP) and trace preserving (TP) quantum map which
acts on the set of density matrices:
\begin{equation}
\Phi:\c M_N\rightarrow\c M_N.
\end{equation}
Complete positivity means that any extended map acting on an enlarged quantum
system
\begin{equation}
\Phi\otimes{\1}_d:\c M_{Nd}\rightarrow\c M_{Nd}
\end{equation}
transforms positive matrices into positive matrices for any extension of
dimension $d$. Due to the Choi theorem, see e.g. \cite{BZ06}, to verify whether a
given quantum map $\Phi^A$ acting on a quantum $N$--level system $A$ is
completely positive it is necessary and sufficient that the following operator
on the Hilbert space $\c H_{N}^A\otimes\c H_{N}^B$  of a composed subsystems $A$
and $B$
\begin{equation}
D_{\Phi^A}:= N(\Phi^A\otimes\mathbbm{1}^{B})\big(|\phi_+^{AB}\left.\right\rangle\left\langle\right. \phi_+^{AB}|\big)\geq 0,
\label{choi}
\end{equation}
is non-negative. Here
$|\phi_+^{AB}\left.\right\rangle=\frac{1}{\sqrt{N}}\sum_{i=1}^N
|i^A\left.\right\rangle\otimes|i^B\left.\right\rangle\in \c H_N^A\otimes\c
H_N^B$ denotes the maximally entangled state in the extended space. The above
relation, called the \emph{Jamio{\l}kowski isomorphism}, implies a
correspondence between quantum maps $\Phi$ and quantum states
$\omega_{\Phi}=\frac{1}{N}D_{\Phi}$. The operator $\omega_\Phi$  is called the
Jamio{\l}kowski--Choi state, whereas the matrix $D_{\Phi}$ is called the
\emph{dynamical matrix} associated with the map $\Phi$.

Any quantum channel acting on a quantum system $A$ can be represented by a
unitary transformation $U^{AB}$ acting on an enlarged system and
 followed by the partial trace over the ancillary subsystem $B$: 
\begin{equation}
\Phi(\rho^A)={\rm Tr}_B\left[U^{AB}(\rho^A\otimes \ket{1^B}\bra{1^B})(U^{AB})^{\dagger}\right].
\label{repenv}
\end{equation}
This formula is called the environmental representation 
of a quantum channel.
Another useful representation of a quantum channel
is given by a set of operators $K_i$ satisfying an  
identity resolution, $\sum_iK_i^{\dagger}K_i=\idty$,
which implies the trace preserving property. 
The Kraus operators $K_i$  define the {\sl Kraus representation}
 of the map  $\Phi$,
\begin{equation}
\Phi(\rho)=\sum_iK_i\rho K_i^{\dagger} .
\label{repkraus}
\end{equation}

Since $\Phi:\rho\rightarrow\rho'$ acts on an operator $\rho$, it is sometimes
called a \emph{superoperator}. If we reshape a density matrix into a vector of
its entries $\vec{\rho}$, the  superoperator $\Phi$ is represented by a matrix
of size $N^2$. It is often convenient to write the discrete dynamics
$\vec{\rho'}=\Phi \vec{\rho}$ using the four--index notation
\begin{equation}
\label{super}
	\rho'_{k\ \!\!l}=\Phi_{\!\!\stackidx{k}{l}{m}{n}}
	\rho_{m\ \!\!n},
	\end{equation}
where the sum over repeating indices is implied and
\begin{equation}
\Phi_{\!\!\stackidx{k}{l}{m}{n}}
=\<k\ l|\Phi|m\ n\> .
\end{equation}
The dynamical matrix $D=D_{\Phi}$
is related to the superoperator matrix $\Phi$
by reshuffling its entries,
\begin{equation}
D_{\stackidx{k}{m}{l}{n}}
=\Phi_{\!\!\stackidx{k}{l}{m}{n}}
\label{reshuffling}
\end{equation}
written $D_{\Phi}=\Phi^R$ or $\Phi= D_{\Phi}^R$.

The von Neumann entropy of the Jamio{\l}kowski-Choi state was studied in
\cite{verstaete,roga0,ziman,RFZ11,Ra12} and it is also investigated in this work. We
are going to compare the spectral properties of the Jamio{\l}kowski-Choi state
$\omega_{\Phi}$ and the spectral properties of the corresponding superoperator
matrix $\Phi$.

\subsection{Entropy of a map}

Entropy $S^{\map}(\Phi)$ is defined \cite{ZB04} as the von
Neumann entropy of the corresponding Jamio{\l}kowski-Choi state
$\omega_{\Phi}=\frac{1}{N}D_\Phi$,
\begin{equation}
\label{smapphi}
S^{\map}(\Phi)\; :=\; -\tr \omega_{\Phi} \ln \omega_{\Phi}.
\end{equation}
This quantity can be interpreted as the special case of the \emph{ exchange entropy}
\cite{szuma}
\begin{equation}
S^{\,\rm exchange}\left(\Phi^A,\rho^A\right)\equiv S\left(\Phi^A\otimes\idty^B(|\phi^{AB}_{\rho^A}\>
\<\phi^{AB}_{\rho^A}|)\right),
\label{exchange}
\end{equation}
where $|\phi^{AB}_{\rho^A}\>\in\c H_{N_A}^A\otimes\c H_{N_B}^B$ is a
purification of $\rho^A$, that is such a pure state of an enlarged system which has the partial trace 
given by ${\rm Tr}_B|\phi^{AB}_{\rho^A}\>\<\phi^{AB}_{\rho^A}|=\rho^A$.
The exchange entropy characterizes the information exchanged during a quantum
operation between a principal quantum system $A$ and an environment $B$, assumed
to be initially in a pure state. Under the condition that an initial state of the
quantum system $A$ is maximally mixed, $\rho^A_*=\frac{1}{N}\idty$, the exchange
entropy $S^{\,\rm exchange}\left(\Phi^A,\rho^A_*\right)$ is equal to the entropy
of a channel $S^{\map}\left(\Phi^A\right)$.

We will treat the entropy of a map as a measure of disturbance caused by a measurement performed on the quantum system.
The work \cite{Maccone2007} contains a list of the properties expected from a good measure of disturbance. 
Among them there is the requirement that the disturbance measure should be equal to zero if and only if
the measuring process is invertible. For unitary transformations of the quantum state
the dynamical matrix  given in Eq. (\ref{choi}) has rank one, 
so  the related entropy of the map is  equal to zero as expected.
Moreover, if the map preserves identity,
the entropy of a map is equivalent to the 
state independent disturbance analyzed in \cite{Maccone2007}.

It is useful to generalize the von Neumann entropy and to
introduce the family of the  R{\'e}nyi entropies
\begin{equation}
\label{ren1}
S_q(\rho)=\frac{1}{1-q}\ln\tr\rho^{q},
\end{equation}
as they allow to formulate a more general
class of uncertainty relations \cite{BB06}.
Here  $q\geq 0$ is a free parameter
and in the limit $q\rightarrow 1$
the generalized entropy tends to the von Neumann entropy,
$S_q(\rho)\rightarrow S_1(\rho)\equiv S(\rho)$.
For any classical probability vector and any
quantum state  $\rho$ the R{\'e}nyi entropy $S_q(\rho)$
is a monotonously decreasing function of the
R{\'e}nyi parameter $q$ \cite{beck}.
The generalized entropy (\ref{smapphi})
of a quantum map $\Phi$, obtained by applying the above form of R{\'e}nyi
to the state $\omega_{\Phi}$ will be denoted by $S_q^{\map}(\Phi)$.

\subsubsection{Connection with the uncertainty principle for measurements}

Let $\Phi$ be a CP TP map with Kraus operators $\{A_i\}_i$, we define, $P_i =
A_i^\dagger A_i$ and note that the operators $\{P_i\}_i$ form a POVM~\cite{BZ06},
i.e. are  positive semidefinite and
\begin{equation}
\sum_i P_i = \1.
\end{equation}
If $\rho$ is a state of a given system, the probability of the outcome
associated with a measurement of the operator $P_i$ reads
\begin{equation}
p_i = \tr P_i \rho.
\end{equation}
The uncertainty involved in a described measurement can be quantified by the entropy~\cite{KP02}
\begin{equation}
H_q(P,\rho) = S_q(p).
\end{equation}
Let us consider an uncertainty involved in the measurement in the case when the state of a given
system is maximally mixed, i.e.
\begin{equation}
p_i = \tr \left(P_i \frac1N \1\right) = \frac1N \tr P_i.
\end{equation}
We have the following corollary
\begin{corollary}\label{cor:min-POVM}
If the state of a given system is maximally mixed, then
\begin{equation}
\min_{P} H_q\left(P,\frac1N \1 \right) = S_q^{\mathrm{map}} (\Phi),
\end{equation}
where the minimum is taken over all possible POVM's such that
\begin{equation}
P_i = A_i^\dagger A_i
\end{equation}
and $A_i$ are the Kraus operators of the quantum channel $\Phi$.
\end{corollary}
\begin{proof}
Let $A_i$ be a Kraus representation of the channel $\Phi$ and denote by $\ket{res(A_i)}$  a vector obtained from the matrix $A_i$ by putting its elements in the lexicographical order i.e. rows follow one after another. We introduce
\begin{equation}
\begin{split}
\kappa_i &= \scalar{res(A_i)}{res(A_i)} = \tr A_i^\dagger A_i = \tr P_i, \\
\ket{a_i} &= \frac{1}{\sqrt{\kappa_i}} \ket{res(A_i)}.
\end{split}
\end{equation}
Assume that we put the $\kappa_i$ coefficients in a decreasing order such that $\kappa_1$ has the largest value. We have~\cite{BZ06}
\begin{equation}
D_{\Phi} = \sum_{i=1}^{l} \ketbra{res(A_i)}{res(A_i)},
\end{equation}
where $l\leq N^2$. Using the variational characterization of eigenvalues~\cite{HJ}, for the
Hermitian matrix $D_{\Phi}$, we get for each $k\leq N^2$ the following expression for the sum of $k$ largest eigenvalues
\begin{eqnarray}
\sum_{i=1}^{k}\lambda_{i}\left(D_{\Phi}\right) & = & \max_{U_{k}}\textrm{Tr}\left(U_{k}^{\dagger}D_{\Phi}U_{k}\right)\\
 & = & \max_{U_{k}}\sum_{i=1}^{l}\kappa_{i}\textrm{Tr}\left(U_{k}^{\dagger}\left|a_{i}\right\rangle \left\langle a_{i}\right|U_{k}\right), \nonumber 
\end{eqnarray}
where $U_{k}$ is a matrix of size $N^{2}\times k$ fulfilling the
relation $U_{k}^{\dagger}U_{k}=\1_k$, and $\1_k$ denotes the $k\times k$ identity. 
For a specific choice $\tilde{U}_k$ of the matrix $U_k$, such that the vectors
$\ket{res(A_1)},\ldots,\ket{res(A_k)}$ belong to the subspace spanned by all $k$
columns of $\tilde{U}_k$, we have $\textrm{Tr} \left(\tilde{U}_{k}^{\dagger}
\left|a_{i}\right\rangle \left\langle a_{i}\right|\tilde{U}_{k}\right)=1$, for
$i=1,\ldots,k$. This property implies 
\begin{equation}
\begin{split}
\sum_{i=1}^{k}\lambda_{i}\left(D_{\Phi}\right)&=
\max_{U_{k}}\sum_{i=1}^{l}\kappa_{i}\textrm{Tr}\left(U_{k}^{\dagger}\left|a_{i}\right\rangle \left\langle a_{i}\right|U_{k}\right)\\
&\geq\sum_{i=1}^{l}\kappa_{i}\textrm{Tr}\left(\tilde{U}_{k}^{\dagger}\left|a_{i}\right\rangle \left\langle a_{i}\right|\tilde{U}_{k}\right) \\
&=\sum_{i=1}^{k}\kappa_{i}+\sum_{i=k+1}^{l}\kappa_{i}\textrm{Tr}\left(\tilde{U}_{k}^{\dagger}\left|a_{i}\right\rangle \left\langle a_{i}\right|\tilde{U}_{k}\right)\\
&\geq\sum_{i=1}^{k}\kappa_{i} .
\end{split}
\end{equation}
In the last inequality we neglected the remaining non-negative terms
labeled by $i>k$.
 The above set of inequalities imply
the majorization relation,  $\kappa \prec \lambda(D_\phi)$. Using the fact that
R{\'e}nyi entropies are Schur--concave, we arrive at the desired inequality,
\begin{equation}
\begin{split}
H_q\left(P,\frac1N \1 \right) &= S_q\left(\left\{\tr P_i \frac1N \1\right\}\right) = S_q(\kappa/N)\\
& \geq S_q(\lambda(D_\Phi)/N) = S_q^{\mathrm{map}}(\Phi).
\end{split}
\label{eq21}
\end{equation}
\end{proof}
Using the monotonicity of the R{\'e}nyi entropies, we get
\begin{equation}
S_q^{\mathrm{map}} (\Phi) \geq S_\infty^{\mathrm{map}} (\Phi) = - \log (\lambda_1(D_{\Phi})/N).
\end{equation}
For $q=1$ this inequality combined with (\ref{eq21}) resembles the
uncertainty principle for a single quantum measurement \cite{KP02,CYGG11,CCYZ12},
since for an optimal POVM the lower bound obtained in these papers
depends on $c=\max_j \tr P_j = \lambda_1/N$.

\subsection{Receiver entropy}

Since a superoperator matrix $\Phi$ is in general not Hermitian, we characterize
this matrix by means of the entropy of the normalized vector of its singular
values
\begin{equation}
S^{\rec}(\Phi)=-\sum_i \mu_i \ln \mu_i .
\label{receiver}
\end{equation}
Here $\mu_i=\frac{\sigma_i}{\sum_k\sigma_k}$ and $\sigma_i$ denote the singular
values of $\Phi$, so that $\sigma_i^2$ are eigenvalues of the positive matrix
$\Phi\Phi^{\dagger}$. This quantity characterizes an arbitrary matrix $\Phi$ and
depends only on its singular values, so it was called \emph{singular quantum
entropy} by Jumarie \cite{Ju00}.

In the case of one--qubit states entire set $\c M_2$ can be represented as a
three-dimensional ball of the radius one - Bloch ball. For one qubit
bistochastic channels,  which preserve the center of the Bloch ball,  the
entropy $S^{\rec}(\Phi)$ characterizes the vector $\{1,\eta_1,\eta_2,\eta_3\}$
after normalization to unity, where  $\{\eta_1,\eta_2,\eta_3\}$  denote the
lengths of semiaxes of an ellipsoid obtained as the image of the Bloch ball
under a quantum operation $\Phi$ -- see \cite{fuji,RSW02}
 and Fig. \ref{elipsoida}.
As demonstrated with a few examples presented below
the entropy $S^{\rec}$ characterizes
the a priori knowledge of the receiver of the outcome of a quantum channel $\Phi$. Therefore the
quantity (\ref{receiver}) will be called the \emph{receiver entropy}.
\begin{figure}
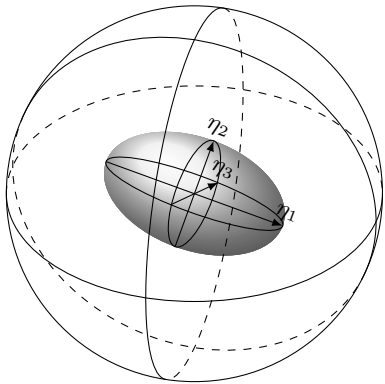

\ellipsoidInBall
\caption{(Color online) An image of the Bloch ball under an exemplary
one--qubit bistochastic map $\Phi$.}
\label{elipsoida}
\end{figure}

To illustrate the meaning of the receiver entropy let us consider the
following examples of one-qubit channels. In the case of completely
depolarizing channel $\Phi_*:\rho\rightarrow\rho_*=\frac{1}{N}\idty$ the entire
Bloch ball is transformed into a single point. All three semiaxes vanish,
 $\eta_i=0$, so $\mu=(1,0,0,0)$ and
the receiver entropy is equal to $0$.
 This value characterizes the perfect knowledge of the receiver of the states
 transmitted through the channel, since having the information about this channel
the receiver knows  that every time he or she obtains
 the very same output state.

In the case of a coarse graining channel $\Phi_{CG}$, which
sets all off--diagonal elements of a density matrix to zero,  and preserves the
diagonal populations unaltered,
the Bloch ball is transformed into an interval of unit length.
This means that the information about
the possible output state missing to the receiver
can be described by a single variable only.
In this case one has $\mu=(1/2,1/2,0,0)$,
so the entropy reads $S^{\rec}(\Phi_{CG})= \ln 2$.
Consider now an arbitrary unitary channel, which only rotates the entire Bloch ball.
Then the receiver has no knowledge in which part of the Bloch ball the output state will
appear, and the corresponding receiver entropy is maximal, $S^{\rec}=2\ln 2$.

In general, the receiver entropy is bounded by the logarithm of the rank of the
superoperator characterizing the channel $S^{\rec}(\Phi)\leq \ln{\rm rank}
(\Phi)$. Every quantum state can be represented by a real vector in the basis of
generalized Pauli matrices (see for instance \cite{pittenger}). Therefore, in
this basis the superoperator is a real matrix and its rank characterizes the
dimensionality of the vector space accessible for the outcomes from the channel.

Consider any orthonormal basis $\{K_i\}$ with respect to the Hilbert-Schmidt scalar product
which includes rescaled identity. 
Such a set of matrices satisfies normalization condition $\sum_iK_i^{\dagger}K_i=\idty$,
therefore can define the POVM measurement.
During the measurements of a quantum state $\rho$ 
the outcomes $K_i\rho K_i^{\dagger}/(\tr K_i\rho K_i^{\dagger})$ 
are observed with probabilities $p_i=\tr K_i\rho K_i^{\dagger}$.
The receiver entropy is related to the probability
distribution characterizing frequency of different outcomes of the measurement apparatus. 
If the entropy is low the receiver may expect
that  only a small amount of outcomes of the measuring apparatus will occur.
High values of the entropy imply that several different results of the measurement
will appear. 
Hence the receiver entropy $S^{\rec}$ characterizes the number of
measurement operators needed to obtain a complete information about the measured state.

\section{One qubit examples}

To analyze discrete dynamics of a one--qubit system let us define two matrices
of order four:
\begin{equation}
G=
\left(\begin{smallmatrix}
1 &0 &0 &0 \\
0 &1 &0 &0 \\
0 &0 &1 &0 \\
0 &0 &0 &1
\end{smallmatrix}
\right)
=
\left(\begin{smallmatrix}
G_1 & | & G_2 \\
 -- & + & --  \\
G_3 & | & G_4 \\
\end{smallmatrix}
\right),
\label{blo}
\end{equation}
and
\begin{equation}
C=G^R=
\left(\begin{smallmatrix}
1 &0 &0 &1 \\
0 &0 &0 &0 \\
0 &0 &0 &0 \\
1 &0 &0 &1
\end{smallmatrix}
\right).
\end{equation}
Note that the first row of $C$ is obtained by reshaping the block $G_1$ into a
vector, the second row of $C$ contains the reshaped block $G_2$, etc. Such
transformation of a matrix is related to the fact that in any linear, one--qubit
map, $\rho'=\Phi\rho$, the $2 \times 2$  matrix $\rho$ is treated as a vector of
length $4$.

Normalizing the spectra of both matrices to unity we get the entropies
$S(G)=\ln{4}$ and $S(C)=0$. Making use of the above notation we can represent
the identity map $\Phi_{\idty}=G$ and the corresponding dynamical matrix
$D_{\idty}=\Phi_{\idty}^R=C$. Moreover, the completely depolarizing channel
which maps any state $\rho$ into the maximally mixed state
$\Phi_{*}:\rho\rightarrow\frac{1}{2}\idty$ can be written as
$\Phi_{*}=\frac{1}{2}C$ while $D_{*}=\Phi_{*}^R=\frac{1}{2}G$. In both cases the
sum of the entropy of a dynamical matrix $S^{\map}\equiv S(\frac{1}{2}D)$ and
the entropy of normalized singular values of superoperator $S^{\rec}\equiv
S(\frac{|\Phi|}{\tr{|\Phi|}})$ reads $S^{\map}+S^{\rec}=2\ln{2}$. Thus, both maps
$\Phi_{\idty}$ and $\Phi_{*}$ are in a sense distinguished, as they occupy
extreme positions at both entropy axes. It is easy to see that the above
reasoning can be generalized for an arbitrary dimension $N$.  For the identity
map acting on $\c M_N$ and the maximally depolarizing channel $\Phi_{*}$ one
obtains $S^{\map}+S^{\rec}=2\ln{N}$. For these two maps the above relation holds
also for the R{\'e}nyi entropies, $S_q^{\map}+S_q^{\rec}=2\ln{N}$.

Investigations of one--qubit quantum operations enabled us to specify the set of
admissible values of the channel entropy $S^{\map}(\Phi)$ and the receiver
entropy $S^{\rec}(\Phi)$. 
We analyzed the images of the set of one--qubit quantum maps 
on to the plane   $(S^{\map}, S^{\rec})$. 
This problem  was first analyzed numerically
by constructing random one-qubit maps \cite{frob}
and marking their position on the plane. 
A special care was paid to the case of bistochastic maps, i.e. maps preserving the identity,
which form a tetrahedron spanned by the identity $\sigma_0=\idty$ and the three Pauli 
matrices $\sigma_i$ (see e.g. \cite{BZ06})
\begin{equation}
\Phi_{bist}(\rho)=\sum_{i=0}^3p_i\sigma_i\rho \sigma_i.
\label{paulichan}
\end{equation}

Fig. \ref{fig:sketch} can thus be interpreted as a
non--linear projection of the set of all one--qubit channels onto the plane
$(S^{\map},S^{\rec})$,
in which bistochastic maps correspond to the dark stripped region.

\begin{figure}[ht]
\centering
\scalebox{.4}{\includegraphics{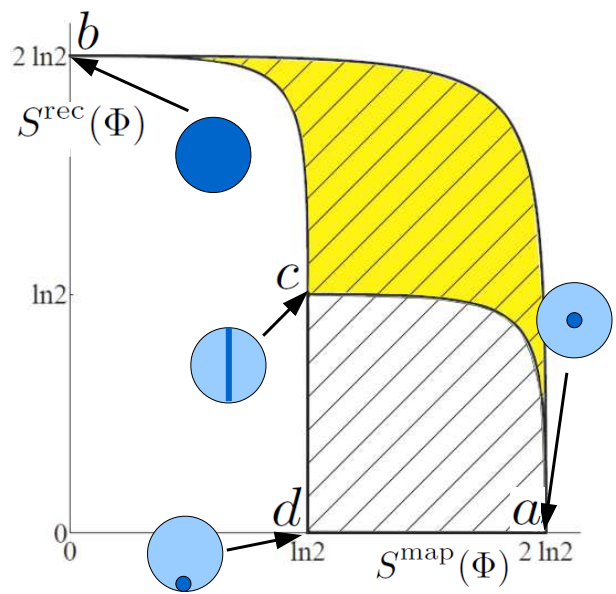}}
\caption{(Color online) 
Striped region denotes the allowed set of points representing one--qubit quantum
operations characterized by the entropy of the channel $S^{\map}(\Phi)$ and the
receiver entropy $S^{\rec}(\Phi)$. Gray (colored) region represents
bistochastic channels, while white stripped region corresponds to interval
channels. Distinguished points of the allowed set represent: $a$ --- completely
depolarizing channel, $b$ --- identity channel, $c$ --- coarse graining channel
and $d$ --- channel completely contracting the entire Bloch ball into a given
pure state.  The action of these four channels on the Bloch ball is
schematically shown on the auxiliary circular plots.
} \label{fig:sketch}
\end{figure}

\medskip

The distinguished points of the allowed region in the plane
$(S^{\map},S^{\rec})$ correspond to:
\begin{itemize}
\item $a)$ completely depolarizing channel:
$\Phi_*:\rho\rightarrow\rho_*=\frac{1}{2}{\1}_2$ ,
\item $b)$ identity channel  $\Phi_{\idty}={\1}$,
\item $c)$ coarse graining channel $\Phi_{CG}$, which sets all off--diagonal
elements of a density matrix to zero, and preserves the diagonal populations
unaltered,
\begin{equation}
\Phi_{CG}=
\left(
\begin{smallmatrix}
1 &0 &0 &0 \\
0 &0 &0 &0 \\
0 &0 &0 &0 \\
0 &0 &0 &1
\end{smallmatrix}
\right),
\end{equation}
\item $d)$ spontaneous emission channel sending any state into a certain pure
state  (e.g. the ground state of the system),
$\Phi_{SE}:\rho\rightarrow|0\rangle\langle0|$.
\end{itemize}

Basing on the numerical analysis the following curves are recognized as the limits of the region
\begin{itemize}
\item the curve $ab$ given by the depolarizing channels
$\Phi_{\alpha}=\alpha\idty+(1-\alpha)\Phi_*$, for $0\leq \alpha\leq 1$.
This family of states provides
the upper bound for the entire region of $(S^{\map},S^{\rec})$  available for the one-qubit quantum maps.
Because of the importance of this curve we provide its parametric expression
\begin{eqnarray*}
\ \ \ \ \  S^{\map}\!&=&\!-\frac{3}{4}(1\!-\!\alpha)\ln\!\left[\frac{1}{4}(1\!-\!\alpha)\right]\!-\!(1\!+\!3\alpha)\ln(1\!+\!3\alpha),\\
S^{\rec}&=&\ln(1+3\alpha)-\frac{3\alpha\ln\alpha}{(1+3\alpha)}.
\end{eqnarray*}
\item the curve $bc$ which represents the  combination of identity and the
coarse graining,
\item the interval $ad$ which represents completely contracting channels:
linear combinations of the completely
depolarizing channel and the spontaneous emission,
\item the interval $cd$ which includes the maps of the form:
\begin{equation}
\Phi_{cd}=
\left(
\begin{smallmatrix}
\alpha &0 &0 &\beta \\ \
\sqrt{\alpha(1-\alpha)}e^{i\phi_1} &0 &0 &\sqrt{\beta(1-\beta)}e^{i\phi_2} \\
\sqrt{\alpha(1-\alpha)}e^{-i\phi_1} &0 &0 &\sqrt{\beta(1-\beta)}e^{-i\phi_2} \\
1-\alpha &0 &0 & 1-\beta
\end{smallmatrix}
\right)
\label{czysty}
\end{equation}
\end{itemize}
with $\alpha, \beta \in (0,1)$
and two arbitrary phases $\phi_1$ and $\phi_2$.

The above maps belong to a broader family of one--qubit operations:
\begin{equation}
\Phi_I=
\left(
\begin{smallmatrix}
\alpha &0 &0 &\beta \\ \
\gamma_1 &0 &0 &\gamma_2 \\
\bar{\gamma}_1 &0 &0 &\bar{\gamma}_2 \\
1-\alpha &0 &0 & 1-\beta
\end{smallmatrix}
\right),
\end{equation}
with complex numbers $\gamma_1$ and $\gamma_2$,
such that the first column
reshaped into a matrix of order two forms a positive state $\rho_1$,
while the reshaped
last column corresponds to a state $\rho_2$.

These operations can be called \emph{interval channels}, as they transform the
entire Bloch ball into an interval given by the convex combination of the states
$\rho_1$ and $\rho_2$. The dynamical matrix corresponding to an interval map can
be transformed by permutations into a block diagonal form.

Fig. \ref{fig:sketch} representing all one-qubit channels
distinguishes two regions. Bistochastic quantum
operations
correspond to the dark region. 
As the set of  one--qubit bistochastic maps forms a tetrahedron (a convex set given in Eq. (\ref{paulichan}))
spanned by three Pauli matrices and identity (see \cite{BZ06}),
to justify this observation that the bistochastic maps cover the dark region of Fig. \ref{fig:sketch}  it is sufficient to analyze the images
of the edges of the antisymmetric part of the tetrahedron onto the plane
$(S^{\map},S^{\rec})$.
The white striped region $acd$ contains for instance interval maps, which will be
shown in Proposition \ref{propozition1}. Note that there exist several maps
which correspond to a given point in Fig. \ref{fig:sketch}.

A further insight into the interpretation of the receiver entropy is due to the
fact that for any completely contractive channel (interval $ad$ in the plot),
which sends any initial state
into a concrete, selected  state, $\Phi_{\xi}:\rho\rightarrow\xi$,  the receiver
entropy is equal to zero. This is implied by the fact that the dynamical matrix of
such an operation reads $D_{\Phi_{\xi}}=\xi\otimes\idty$. After reshuffling of
this matrix we obtain the superoperator matrix of rank one, since all non-zero
column are the same, therefore it has only one nonzero singular value.
Normalization of the vector of singular values sets this number to unity so that
$S^{\rec}(\Phi_{\xi})=0$. This observation  supports an interpretation of
$S^{\rec}$ as the amount of information missing to the receiver of the output
$\rho'$  of a quantum channel, who  knows the operation $\Phi$,  but does not
know the input state~$\rho$.

\section{Entropic inequalities for reordered matrices}
Before we establish several 
trade--off
relations for quantum channels we shall introduce a  framework concerning matrices 
(in general non--hermitian) together with their reordered counterparts. 
An arbitrary $d\times d$ matrix $X$ has $d^2$ independent matrix elements. A matrix $Y_\pi$ can be called a reordering of $X$ if $Y_\pi=X^\pi$, where $\pi$ denotes some permutation of matrix entries. Thus, for each matrix $X$ we can consider $\left(d^2\right)!$ reordered matrices $Y_\pi$.

Denote by $x_{i}$ the singular values of the matrix $X$ and introduce the following $q$-norms:
\begin{equation}\label{qnorms}
\left\Vert X\right\Vert _q=\left(\textrm{Tr}\left[XX^{\dagger}\right]^{q/2}\right)^{1/q}=\left(\sum_{i}x_{i}^{q}\right)^{1/q}.
\end{equation}
Moreover, by $x_{1}\equiv\left\Vert X\right\Vert _{\infty}$ denote the
greatest singular value of the matrix $X$ and by
\begin{equation}
\Lambda_{x}=\left\Vert X\right\Vert _{1}=\sum_{i}x_{i},
\end{equation}
the trace norm of $X$, i.e. the sum of all singular values $x_i$. Finally, define the R{\'e}nyi
entropy
\begin{equation}
S_{q}\left(X\right)=\frac{1}{1-q}\ln\sum_{i}\left(\frac{x_{i}}{\left\Vert X\right\Vert _{1}}\right)^{q}.
\end{equation}
The first result holds in general. \begin{lemma}\label{lemma2}
For an arbitrary matrix $X$ and $1\leq q<\infty$ we have
\begin{equation}
\ln\left(\frac{\Lambda_{x}}{x_{1}}\right)\leq S_{q}\left(X\right)\leq\frac{q}{q-1}\ln\left(\frac{\Lambda_{x}}{x_{1}}\right),\label{boundy1}
\end{equation}
\end{lemma}
The second inequality relates matrices $X$ and $Y_\pi$. \begin{lemma}\label{lemma3}
If $Y_\pi=X^{\pi}$ where the transformation $\pi$ is an arbitrary permutation
of matrix entries, then we have for $1\leq q<\infty$
\begin{equation}
F_{\textrm{min}}\ln\left(\frac{\Lambda_{y}}{\sqrt{x_{1}\Lambda_{x}}}\right)\leq S_{q}\left(Y_\pi\right)\leq F_{\textrm{max}}\ln\left(\frac{\Lambda_{y}}{x_{1}}\right),\label{boundyb}
\end{equation}
where $F_{\textrm{min}}=\min\left(\frac{q}{q-1};2\right)$ and $F_{\textrm{max}}=\max\left(\frac{q}{q-1};2\right)$.\end{lemma}
 The symbol $\Lambda_y$ inside Lemma \ref{lemma3} denotes the trace norm of $Y_\pi$. Both lemmas are proven in Appendix \ref{algebraic lemmas}.

\section{Trade--off relations for quantum channels}\label{uncer}
The structure of the set of allowed values of both entropies $S^{\map}$ and
$S^{\rec}$ describing all one--qubit stochastic maps shown in Fig.
\ref{fig:sketch} suggests that their sum is bounded from below. For the smaller
class of bistochastic maps the bound looks to be more tight. Indeed we are going
to prove the following trade--off relation for the sum of two von Neumann entropies
\begin{subequations}
\begin{equation}
 S^{\map}(\Phi)+ S^{\rec}(\Phi)\geq \ln{N},
\label{gen_1}
\end{equation}
and a sharper inequality
\begin{equation}
 S^{\map}(\Phi)+ S^{\rec}(\Phi)\geq 2\ln{N},
\label{bist_1}
\end{equation}
\end{subequations}
which holds for any bistochastic map acting on a $N$ dimensional system. Note that
the second expression can be interpreted as a kind of entropic trade--off
relation for unital quantum channels: if the map entropy $S^{\map}(\Phi)$, which
quantifies the interaction with the environment during the operation or the degree of disturbance of a quantum state, is small,
the receiver entropy $S^{\rec}(\Phi)$ cannot be small as well. This implies that
the results of the measurement could be very diverse. Conversely, a small value
of the receiver entropy implies that the map $\Phi$ is strongly contracting,  
so the map entropy is sufficiently large and a lot of information escapes from the
system to an environment and the
disturbance of the initial state is strong.

\medskip
Instead of proving directly the bounds (\ref{gen_1}) and (\ref{bist_1}) for the von Neumann entropy
$S\equiv S_1$ we are going to prove a more general inequalities formulated for the
R{\'e}nyi entropies $S_q$ with $q \in [1,\infty[$. All bounds in the limiting case $q=2$,
related to the Hilbert--Schmidt  norm of a matrix, are shown in Fig.
\ref{fig:sketch2} for one--qubit quantum  operations.

In the case of the $N^2 \times N^2$ matrices $\Phi$ and $D_{\Phi}$, let us denote by
$\sigma_{1}$ the greatest singular value of $\Phi$, by $d_{1}$
the greatest eigenvalue of $D_{\Phi}$, and by $\Lambda_{\Phi}=\left\Vert \Phi\right\Vert _{1}$
the sum of all singular values of $\Phi$. Since the Jamio{\l}kowski--Choi state $\omega_\Phi$ is normalized we have
$\left\Vert D_{\Phi}\right\Vert _{1}\equiv N$.

If we apply Lemma \ref{lemma2} to both matrices we obtain the following bounds:
\begin{subequations}
\begin{equation}
\ln\left(\frac{\Lambda_{\Phi}}{\sigma_{1}}\right)\leq S_q^{\rec}(\Phi)\leq\frac{q}{q-1}\ln\left(\frac{\Lambda_{\Phi}}{\sigma_{1}}\right),\label{boun1}
\end{equation}
\begin{equation}
\ln\left(\frac{N}{d_{1}}\right)\leq S_q^{\map}(\Phi)\leq\frac{q}{q-1}\ln\left(\frac{N}{d_{1}}\right).\label{boun2}
\end{equation}
Since $\Phi=D_\Phi^R$ and $D_\Phi=\Phi^R$,
 where the reshuffling operation $R$ is a particular example of reordering we have
 two additional bounds  originating from Lemma \ref{lemma3}:
\begin{equation}
F_{\textrm{min}}\ln\left(\frac{\Lambda_{\Phi}}{\sqrt{Nd_{1}}}\right)\leq  S_q^{\rec}(\Phi)\leq F_{\textrm{max}}\ln\left(\frac{\Lambda_{\Phi}}{d_{1}}\right),\label{boun3}
\end{equation}
\begin{equation}
F_{\textrm{min}}\ln\left(\frac{N}{\sqrt{\sigma_{1}\Lambda_{\Phi}}}\right)\leq S_q^{\map}(\Phi)\leq F_{\textrm{max}}\ln\left(\frac{N}{\sigma_{1}}\right).\label{boun4}
\end{equation}
\end{subequations}
The bounds (\ref{boun3}, \ref{boun4}) are in fact implied by the equality of Hilbert--Schmidt norms $\left\Vert \Phi\right\Vert _2=\left\Vert D_\Phi\right\Vert _2$, what is a consequence of the reshuffling relation $D_\Phi=\Phi^R$ .

Inequalities (\ref{boun1}---\ref{boun4}) provide individual limitations for ranges of the entropies $S_q^{\rec}$ and $S_q^{\map}$. However, if we consider a particular inequality we can always recover a full range $[0,2\ln N]$. The above inequalities can be combined in four different ways: (\ref{boun1}) with (\ref{boun2}),  (\ref{boun3}) with (\ref{boun4}),  (\ref{boun1}) with (\ref{boun4}) and  (\ref{boun2}) with (\ref{boun3}) in order to obtain upper and lower bounds for the sum $S_q^{\map}+S_q^{\rec}$. These bounds shall depend on the three parameters: $\sigma_1$, $d_1$ and $\Lambda_\Phi$, thus, without an additional knowledge about these parameters, they do not lead to a trade--off relation. In particular, for $d_1=\sigma_1=\Lambda_\Phi=N$ we find from (\ref{boun1}---\ref{boun4}) that  $S_q^{\rec}=0$ and $S_q^{\map}=0$. This case would correspond
to a pure, separable Jamio{\l}kowski--Choi state $\omega_\Phi$.

In order to show that the above example cannot be realized by a CP TP map we shall prove the following theorem which provides an upper bound on the greatest singular value $\sigma_1$.
\begin{theorem} \label{th:super-bound}
Let $\Phi$ be a CP TP channel acting on a set of density operators of size $N$.
Its superoperator $\Phi$ is a $N^2 \times N^2$ matrix. The greatest singular
value $\sigma_1$ is:

\begin{enumerate}
\item given by the expression
\begin{equation}
\label{singq1}
\sigma_1(\Phi) = \max_{\rho \in \mathcal{M}_N} \sqrt{\frac{ \tr \Phi(\rho)^2}{\tr \rho^2}},
\end{equation}
\item bounded
\begin{equation}
\sigma_{1}\left(\Phi\right)\leq\sqrt{N\tau_1}\leq\sqrt{N},\label{1}
\end{equation}\end{enumerate}where $\tau_1$ denotes the greatest eigenvalue of the density matrix $\Phi\left(\frac{1}{N}\1\right)\in \c M_N$.
\end{theorem}
The bound $\sigma_1\leq\sqrt{N}$ is saturated for quantum channels which transform the maximally
mixed state onto a pure state (only in that case $\tau_1=1$).
In the case of a bistochastic map all eigenvalues of
$\Phi\left(\frac{1}{N}\1\right)$ are equal to $1/N$ and therefore
$\sigma_1(\Phi)\leq 1$.
The proof of Theorem \ref{th:super-bound} is presented in Appendix \ref{greatestsingular}.
Some other bounds on singular values of reshuffled
density matrices have been studied in \cite{chikwongli}. In particular, there was shown that for $\rho \in \c M_N$ the largest singular value of the matrix $\rho^R$ is greater than $N^{-1}$. Because $\Phi=N\omega_\Phi^R$ we immediately find that $\sigma_1\geq 1$. Thus, for bistochastic maps we have the equality $\sigma_1=1$.

We are now prepared to prove the following theorem which establishes the entropic  trade--off relations between
$S_q^{\map}$ and $S_q^{\rec}$.
\begin{theorem} \label{th:general-uncertainty}
For a CP TP map $\Phi$ acting on a system of an arbitrary dimension $N$ the
following relations hold:
\begin{enumerate}
\item For an arbitrary map $\Phi$
\begin{equation}\label{stoch_q}
S_q^{\map}(\Phi)+S_q^{\rec}(\Phi)\geq \frac{F_\textrm{min}}{2}\ln{N},
\end{equation}
\item If the quantum channel $\Phi$ is bistochastic
\begin{equation}\label{bistoch_q}
 S_q^{\map}(\Phi)+ S_q^{\rec}(\Phi)\geq F_\textrm{min}\; \ln{N}.
\end{equation}
\end{enumerate}
\end{theorem}
Since for $q=1$ the coefficient $F_\textrm{min}=2$,
from Theorem \ref{th:general-uncertainty} we recover the particular
bounds (\ref{gen_1}, \ref{bist_1}) for the von Neumann entropies.
\begin{proof}[Proof of Theorem \ref{th:general-uncertainty}] In a first step we shall add two lower bounds present in   (\ref{boun1}) and (\ref{boun4}) to obtain
\begin{equation}\label{bound genL}
S_q^{\map}(\Phi)+S_q^{\rec}(\Phi)\geq F_\textrm{min}\ln\left(\frac{N}{\sigma_1}\right)+\left(1\!-\!\frac{F_\textrm{min}}{2}\right)\!\ln\!\left(\frac{\Lambda_\Phi}{\sigma_1}\right)\!.
\end{equation}
Since $F_\textrm{min}\leq2$ and the greatest singular value $\sigma_1$ is less than the sum $\Lambda_\Phi$ of all singular values, the second term is always nonnegative. Thus, due to the upper bound (\ref{1}) we have
\begin{equation}\label{bound genL2}
S_q^{\map}(\Phi)+S_q^{\rec}(\Phi)\geq \frac{F_\textrm{min}}{2}\ln\left(\frac{N}{\tau_1}\right).
\end{equation}
The first statement of Theorem  \ref{th:general-uncertainty} follows immediately, when instead of $\tau_1$ we put its maximal value $1$ into the inequality (\ref{bound genL2}).
 The second statement is related to the fact that bistochastic
quantum channels preserve the identity i.e. $\Phi\left(\frac{1}{N}\1\right)=\frac{1}{N}\1$. The greatest eigenvalue $\tau_1$ is in this case equal to $\frac{1}{N}$, thus the value of $N^2$ appears inside the logarithm and cancels the factor of $2$ in the denominator.
\end{proof}
In fact, the inequality (\ref{bound genL2}) quantifies the deviation from the set of bistochastic maps, with the greatest eigenvalue of $\Phi\left(\frac{1}{N}\1\right)$ playing the role of the interpolation parameter.
\subsection{Additional upper bounds}
The receiver entropy $S_q^{\rec}(\Phi)$ is upper bounded due to the relations (\ref{boun1}) and (\ref{boun3}). However, these bounds diverge in the limit $q\rightarrow1$. Since the greatest singular value $\sigma_1$ is not less than $1$ we can derive another upper bound which gives a nontrivial limitation valid for all values of $q$.
\begin{theorem} \label{th: recq}
For a CP TP map $\Phi$ acting on a system of an arbitrary dimension $N$ the
following relation holds:
\begin{equation}\label{recq}
 S_q^{\rec}(\Phi)\leq
 \frac{1}{1-q}\ln\left(\Lambda_\Phi^{-q}+\frac{\left(\Lambda_\Phi-1\right)^{q}}
{\Lambda_\Phi^q\left(N^2-1\right)^{q-1}}\right).
\end{equation}
\end{theorem}
\begin{proof}
Since the map $\Phi$ is CP TP the greatest singular value $\sigma_1\geq1$. Thus, the vector $\boldsymbol{\sigma}$ of the singular values of the $N^2\times N^2$ matrix $\Phi$ majorizes ($\boldsymbol{\sigma}\succ\boldsymbol{\sigma}_0$) the vector:
\begin{equation}
\boldsymbol{\sigma}_{0}=\left(1,\underset{N^{2}-1}{\underbrace{\frac{\Lambda_{\Phi}-1}{N^{2}-1},\frac{\Lambda_{\Phi}-1}{N^{2}-1},\ldots,\frac{\Lambda_{\Phi}-1}{N^{2}-1}}}\right).
\end{equation}
Since $S_q^{\rec}(\Phi)= S_q(\boldsymbol{\sigma}/\Lambda_\Phi)$ and the R{\'e}nyi entropy is Schur concave we obtain the inequality $S_q^{\rec}(\Phi)\leq S_q(\boldsymbol{\sigma}_0/\Lambda_\Phi)$ which is equivalent to (\ref{recq}).
\end{proof}

As a limiting case of Theorem \ref{th: recq} we have the corollary
\begin{corollary}\label{upperB}
The von Neumann entropy $S^{\rec}(\Phi)$ is bounded
\begin{equation}
 S^{\rec}(\Phi)\leq \frac{\Lambda_\Phi-1}{\Lambda_\Phi}\ln\left(\frac{N^2-1}{\Lambda_\Phi-1}\right)+\ln\Lambda_\Phi\leq2\ln N.
\end{equation}
\end{corollary}
\begin{figure}[ht]
\centering
\scalebox{.4}{\includegraphics{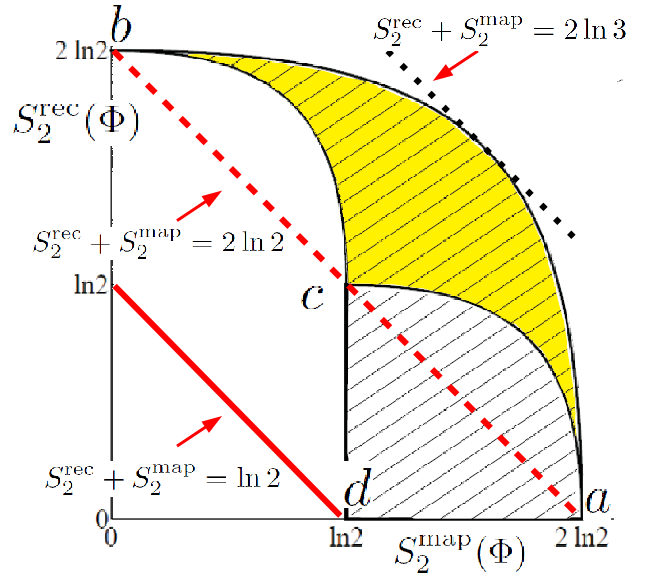}}
\caption{(Color online) 
Stripped region represents the set of one--qubit operations projected into the
plane spanned by the linear entropy of the map, $S_2^{\map}(\Phi)$,  and the
linear receiver entropy $S_2^{\rec}(\Phi)$, i.e. the R{\'e}nyi entropies of order
$q=2$. Dark region represents the bistochastic quantum operations. Dashed
antidiagonal line represents the lower bound (\ref{bistoch_q}) which holds for
bistochastic operations, while solid antidiagonal line denotes the weaker bound
(\ref{stoch_q}) which holds for all quantum operations.
Dotted antidiagonal line represents the upper bound (\ref{upp}) applied for $N=2$.
} \label{fig:sketch2}
\end{figure}

The relation between the matrices $\Phi$ and $D_\Phi$ allows us to derive an upper bound for the sum of the R{\'e}nyi entropies $S_2^{\map}(\Phi) +  S_2^{\rec}(\Phi)$.
\begin{proposition}\label{upperQ1}
The following relation holds:
\begin{equation}\label{upp}
S_2^{\map}(\Phi)+S_2^{\rec}(\Phi)\leq 2\ln\left(\frac{N(N+1)}{2}\right).
\end{equation}
\end{proposition}
\begin{proof}
Since $\left\Vert \Phi\right\Vert _2=\left\Vert D_\Phi\right\Vert _2$ we have an easy relation between both entropies:
\begin{equation}\label{S2}
S_2^{\map}(\Phi)=S_2^{\rec}(\Phi)+2\ln N-2\ln\Lambda_\Phi.
\end{equation}
According to (\ref{recq}) we are able to estimate
\begin{equation}
S_2^{\map}(\Phi)+S_2^{\rec}(\Phi)\leq 2\ln (N\Lambda_\Phi)-2\ln\left(1+\frac{\left(\Lambda_\Phi-1\right)^{2}}{N^2-1}\right).
\end{equation}
In order to complete the proof of Proposition \ref{upperQ1} we shall perform the maximization of the above upper bound over the parameter $\Lambda_\Phi\in\left[1,N^2\right]$.
\end{proof}
The bound presented in Proposition~\ref{upperQ1} can be saturated by a
quantum channel, which is a mixture of the identity channel and the maximally
depolarizing channel, i.e.
\begin{equation}
\Phi = \frac{1}{N+1} \1 + \frac{N}{N+1}\Phi_{*}.
\end{equation}

In fact, we are able to generalize the relation (\ref{S2}) to the case of all $1\leq q\leq\infty$.
\begin{proposition}\label{upperQ2}
The following relation holds:
\begin{equation}\label{upp2}
S_q^{\map}(\Phi)\geq F_\textrm{min}\ln \frac{N}{\Lambda_\Phi} + G_\textrm{min} S_q^{\rec}(\Phi),
\end{equation}
where $G_\textrm{min}=\min\left(\frac{q}{2(q-1)};\frac{2(q-1)}{q}\right)$.
\end{proposition}
\begin{proof}
Assume that $q\leq2$. In that case we have the following monotonicity properties for the R{\'e}nyi entropy: $S_q\geq S_2\geq 2\left(\frac{q-1}{q}\right) S_q$. These relations together with Eq. (\ref{S2}) provide a chain of inequalities:
\begin{eqnarray}
S_{q}^{\textrm{map}}\left(\Phi\right) & \geq & S_{2}^{\textrm{map}}\left(\Phi\right)\nonumber \\
 & = & 2\ln\left(\frac{N}{\Lambda_{\Phi}}\right)+S_{2}^{\textrm{rec}}\left(\Phi\right)\\
 & \geq & 2\ln\left(\frac{N}{\Lambda_{\Phi}}\right)+2\left(\frac{q-1}{q}\right)S_{q}^{\textrm{rec}}\left(\Phi\right),\nonumber
\end{eqnarray}
The same method applied for $q\geq2$ with associated monotonicity relations  $2\left(\frac{q-1}{q}\right) S_q\geq S_2\geq  S_q$
completes  the proof of  inequality (\ref{upp2}).
\end{proof}

\medskip
We can also show in which region of the plot $(S^{\map}, S^{\rec})$ the interval
maps are located. Notice that the classical maps, which transform the set of
$N$--point probability vectors into itself, also satisfy these inequalities.

\begin{proposition}\label{propozition1}
The interval maps satisfy the following inequalities
$S^{\rec}(\Phi)\leq \ln{N} \leq S^{\map}(\Phi)$.
\end{proposition}

\begin{proof}
The left inequality concerning the receiver
entropy follows from the fact that the entire set of states
is mapped into an interval. To show the right inequality
observe that the dynamical matrix corresponding to
an interval channel is block diagonal or can be
transformed to this form by a permutation.
 Due to the trace preserving condition
every block of the normalized dynamical matrix can be interpreted as
$\frac{1}{N}\rho_i$ where $\rho_i$ is some density matrix. Therefore,
 up to a permutation $P$
the normalized dynamical matrix has the structure
\begin{equation}
\omega=\frac{1}{N}P^{\dagger}D_{\Phi}P=\sum_{i=1}^{N}\frac{1}{N}\rho_i\otimes|i\>\<i|.
\end{equation}
Hence the entropy of the normalized dynamical matrix reads
\begin{equation}
\begin{split}
S(\omega)&= S\left(\frac{D_{\Phi}}{N}\right)=-\sum_i\tr\frac{1}{N}\rho_i\ln{\frac{1}{N}\rho_i}\\
&= \ln{N}+\sum_i\frac{1}{N}S(\rho_i).
\end{split}
\end{equation}
This implies the desired inequality for
the entropy of a map, $S^{\map}(\Phi)$.
\end{proof}
The last string of equations exemplifies the Shannon rule known as \emph{the
grouping principle} \cite{shannon, hall} that the information of expanded
probability distribution should be the sum of a reduced distribution and
weighted entropy of expansions. Notice that the grouping rule does not hold for
all dynamical matrices corresponding to generic quantum operations. As an
example take a maximally entangled state which is a purification of the maximally
mixed state.


\subsection{Super--positive maps and separability of the Jamio{\l}kowski--Choi state}

The aim of this part is to answer the question: How the
separability (entanglement) of the state $\omega_\Phi$ 
can be described in terms of the entropies $S_q^{\map}$ and  $S_q^{\rec}$?
In other words we wish to identify the class of {\sl superpositive maps}
(also called entanglement breaking channels -- see \cite{BZ06}),
for which $\omega_\Phi$ is separable on the plane $(S_q^{\map}, S_q^{\rec})$.
Furthermore, we will determine the region on this plane
 where no such maps can be found. The method to answer
these questions is based on the previously given uncertainty relations and
the realignment separability criteria \cite{chenWu,HorRealign}. These criteria
state that if $\omega_\Phi$ is separable then the sum of all singular values of
the matrix $\frac{1}{N}\Phi$ cannot be greater than $1$, what straightforwardly
implies $\Lambda_\Phi\leq N$.
\begin{figure}[ht]
\centering
\scalebox{.4}{\includegraphics{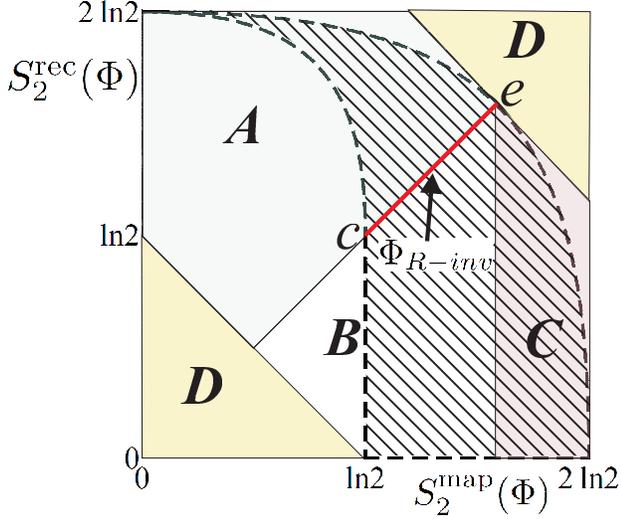}}
\caption{(Color online) One--qubit maps projected onto the entropy plane (stripped set)
with superimposed bounds given in Proposition \ref{separable} concerning the separability of the
dynamical matrix.
Region $A$ contains no superpositive channels,
while region $B$ contains both classes of the maps. Region $C$ (determined by PPT criteria) contains only superpositive channels (all corresponding states are separable).
Lower (\ref{stoch_q})  and upper bounds (\ref{upp}) imply that
there are no one--qubit quantum operations projected into region $D$.
Diagonal of the figure contains the reshuffling--invariant channels,
in particular, the coarse graining channel ($c$)
and the {\sl transition depolarizing channel}
 $\Phi_{1/3}=\frac{1}{3}\idty+\frac{2}{3}\Phi_*$,
located  at the boundary of the set of superpositive maps ($e$).}
\label{fig:entanglement}
\end{figure}
We shall prove the following proposition
\begin{proposition}
\label{separable}
If $\omega_\Phi$ is separable, then:
\begin{enumerate}
\item $S_{q}^{\textrm{map}}\left(\Phi\right)\geq\frac{F_{\textrm{min}}}{4}\ln N$, \;and
\item $S_{q}^{\textrm{rec}}\left(\Phi\right)\leq\frac{1}{1-q}\ln\left(\frac{\left(N+1\right)^{q}+N^{2}-1}{N^{q}\left(N+1\right)^{q}}\right)$, \;and
\item $S_{q}^{\textrm{map}}\geq G_{\textrm{min}}S_{q}^{\textrm{rec}}\left(\Phi\right)$.
\end{enumerate}
\end{proposition}
\begin{proof}
In order to prove the statements 1--3 we apply
the separability criteria $\Lambda_\Phi\leq N$ directly to the inequalities (\ref{boun4}), (\ref{recq}) and (\ref{upp2}) respectively. In the  case 1 we also include the bound $\sigma_1\leq\sqrt{N}$.
\end{proof}
The above result leads immediately to the separability criteria. If at least
one inequality from Proposition \ref{separable} is violated,
then the state $\omega_\Phi$ is entangled,
 so the map $\Phi$ is not superpositive -- see Fig. \ref{fig:entanglement}.

The last inequality in Proposition \ref{separable} is saturated for the channels for which $S_2^{\map}=S_2^{\rec}$. They are located at
the diagonal of Fig. \ref{fig:entanglement}. This class contains maps
with dynamical matrix symmetric with respect to the reshuffling,
$D=D^R=\Phi$. This condition implies that the superoperator $\Phi$
is hermitian so its spectrum is real.
The following proposition characterizes
the set of one--qubit channels invariant with respect
to reshuffling. 

\begin{proposition}
The following one--qubit bistochastic channels
$\Phi_{R-inv}$ are reshuffling--invariant
\begin{equation}
\Phi_{R-inv}=\Phi_U\Phi_{\eta_1,\eta_2}\Phi_{U^{\dagger}},
\label{forma}
\end{equation}
where
\begin{eqnarray}
& &\Phi_{\eta_1,\eta_2}=\frac{1}{2}\begin{pmatrix}1+\eta_3&0&0&1-\eta_3\\
0&\eta_1+\eta_2&\eta_1-\eta_2&0\\
0&\eta_1-\eta_2&\eta_1+\eta_2&0\\
1-\eta_3&0&0&1-\eta_3\label{eta1eta2}\\
\end{pmatrix}\\
& &\qquad\qquad\qquad{\rm and}\quad 1=\eta_1+\eta_2+\eta_3.\nonumber
\end{eqnarray}
The map $\Phi_U=U\otimes \bar{U}$
describes an arbitrary unitary channel,
as $U$ is a unitary matrix of order two
and $\bar{U}$ denotes its complex conjugation.
\end{proposition}
\begin{proof}
To justify the above statement we use the following general property of the reshuffling operation, which can be easily verified by checking the matrix entries of both sides
\begin{eqnarray}
& &\!\!\!\!\!\!\!\!\Big[\left(X^1_{n}\otimes X^2_n\right)\ Y_{n^2}\ \left(X^3_{n}\otimes X^4_{n}\right)\Big]^R=\label{przetas}\\
&=&\left(X^1_{n}\otimes \left(X^3_n\right)^T\right)\ Y_{n^2}^R \ \left(\left(X^{2}\right)^T_{n}\otimes X^4_n\right),\nonumber
\end{eqnarray}
where lower indices denote the dimensionalities of square matrices.
Since (\ref{eta1eta2}) is a reshuffling--invariant matrix, using (\ref{przetas})
we see that (\ref{forma}) is preserved after reshuffling.
\end{proof}

Two extreme examples of the reshuffling--invariant maps
are distinguished in Fig. \ref{fig:entanglement}:
the coarse graining channel ($c$)
 for which $\eta_1=\eta_2=0$ and $\eta_3=1$,
and the transition depolarizing channel ($e$) at the boundary of super--positivity,
$\Phi_{1/3}=\frac{1}{3}\idty+\frac{2}{3}\Phi_*$,
for which  $\eta_1=\eta_2=\eta_3=1/3$.

\section{Concluding remarks}

In this work an entropic trade--off relation analogue of the entropic uncertainty relation characterizing a
given quantum operation (\ref{stoch_q}) was established. We have shown that
for any stochastic quantum map
the sum of the map entropy, characterizing the decoherence introduced to the
system by the measurement process, and the receiver entropy, which describes the
knowledge on the output state without any information on the input,  is bounded
from below.  The more one knows a priori concerning the outcome state, the more
information was exchanged between the principal  subsystem and the environment
due to the quantum operation. A stronger bound (\ref{bistoch_q}) is obtained
for the class of bistochastic maps, for which the maximally mixed state is
preserved. Entanglement properties of a Jamio\l{}kowski--Choi state were
investigated in terms of the entropies $S_q^{\map}$ and  $S_q^{\rec}$.

Dynamical entropic trade--off relations were obtained also 
for the R{\'e}nyi entropies  of an arbitrary order $q$.
From a mathematical perspective this result is
based on inequalities relating the spectrum of a positive hermitian matrix
$X=X^{\dagger}$ and  the singular values of the non--hermitian reshuffled matrix
$X^R$. Related algebraic results were recently obtained in \cite{chikwongli}
and applied to the separability problem. It is tempting to believe that further
algebraic investigations on the spectral properties of a reshuffled matrix will
lead to other results  applicable to physical problems motivated by the quantum
theory.

\begin{acknowledgments}
The authors would like to thank P. Gawron for his help
with the preparation of the figures.
We are grateful to M.~Zwolak and A.E.~Rastegin for helpfull correspondence
and appreciate encouraging discussions with I.~Bia{\l}ynicki--Birula, J.~Korbicz
and  R.~Horodecki. W.R. acknowledges financial support
from the EU STREP Projects HIP, Grant Agreement No. 221889.
Z.P. was supported by MNiSW under the project number IP2011~044271. {\L}.R. acknowledges financial support by MNiSW  research grant, number IP2011~046871, for years 2012-2014.
K.{\.Z}. acknowledges financial support by
the Polish NCN research grant, decision number DEC-2011/02/A/ST2/00305.

\end{acknowledgments}
\appendix

\section{Algebraic lemmas}\label{algebraic lemmas}
In order to prove Lemma \ref{lemma2} we need the following norm inequality:\begin{lemma}\label{lemma1}For
an arbitrary vector $\boldsymbol{x}$ with non-negative coefficients $x_i$,
and for $1\leq q<\infty$ we have
\begin{equation}
\left\Vert \boldsymbol{x}\right\Vert _{q}\leq\left\Vert \boldsymbol{x}\right\Vert _{1}^{1/q}\left\Vert \boldsymbol{x}\right\Vert _{\infty}^{\left(q-1\right)/q}.\label{eq:Lemma1}
\end{equation}
\end{lemma}\begin{proof} For $1/r+1/s=1$ we shall write $x_{i}^{q}=x_{i}^{q/r}x_{i}^{q/s}$ and next use the H{\"o}lder inequality for $1/\alpha+1/\beta=1$
\begin{equation}
\left\Vert \boldsymbol{x}\right\Vert _q\leq\left(\sum_{i}x_{i}^{q\alpha/r}\right)^{1/\left(q\alpha\right)}\left(\sum_{i}x_{i}^{q\beta/s}\right)^{1/\left(q\beta\right)}.
\end{equation}
When we choose $\alpha=\infty$, $\beta=1$, $s=q$ and $r=q/\left(q-1\right)$
we obtain the desired result (\ref{eq:Lemma1}).\end{proof}

\begin{proof}[Proof of Lemma \ref{lemma2}] Lemma \ref{lemma1} together with the fact that the $q$-norms (\ref{qnorms}) are decreasing functions of the $q$ parameter provide a chain of norm inequalities
\begin{equation}\label{chain}
x_{1}=\left\Vert X\right\Vert _{\infty}\leq\left\Vert X\right\Vert _q\leq\left\Vert X\right\Vert _{1}^{1/q}\left\Vert X\right\Vert _{\infty}^{\left(q-1\right)/q}=\Lambda_{x}^{1/q}x_{1}^{\left(q-1\right)/q}.
\end{equation}
We shall divide (\ref{chain}) by $\left\Vert X\right\Vert _{1}\equiv\Lambda_{x}$
to obtain
\begin{equation}
\frac{x_{1}}{\Lambda_{x}}\leq\frac{\left\Vert X\right\Vert _q}{\left\Vert X\right\Vert _{1}}\leq\left(\frac{x_{1}}{\Lambda_{x}}\right)^{\left(q-1\right)/q}.
\end{equation}
When we take the logarithm of the above inequality and then multiply by $q/\left(1-q\right)$,
we boil down to the result (\ref{boundy1}).
\end{proof}
\begin{proof}[Proof of Lemma \ref{lemma3}]
Reordering operations do not change the matrix entries, thus they also do not change the Hilbert--Schmidt norm $\left\Vert \cdot\right\Vert _\mathrm{HS}\equiv\left\Vert \cdot\right\Vert _2$ which is a sum of squares of moduli of all matrix entries.
This implies the equality$\left\Vert X\right\Vert _{2}=\left\Vert Y_\pi\right\Vert _{2}$.

First we shall prove the right hand side of (\ref{boundyb}). For
$1\leq q\leq2$ we have
\begin{equation}
\left\Vert Y_\pi\right\Vert _q\geq\left\Vert Y_\pi\right\Vert _{2}=\left\Vert X\right\Vert _{2}\geq x_{1},
\end{equation}
what by the same steps as before transforms into
\begin{equation}
 S_{q}\left(Y_\pi\right)\leq\frac{q}{q-1}\ln\left(\frac{\Lambda_{y}}{x_{1}}\right).
\end{equation}
For $q\geq2$ we extend the above inequality using the monotonicity property of the R{\'e}nyi entropy $S_{q}\leq S_{2}$.
Finally, we introduce the function $F_{\textrm{max}}=\max\left(\frac{q}{q-1};2\right)$ to describe properly the transition from $1\leq q\leq2$ to $q\geq2$.

In the case of the lower bound  (\ref{boundyb}) we have for $q\geq2$
\begin{equation}
\left\Vert Y_\pi\right\Vert _q\leq\left\Vert Y_\pi\right\Vert _{2}=\left\Vert X\right\Vert _{2}\leq\sqrt{\left\Vert X\right\Vert _{1}\left\Vert X\right\Vert _{\infty}}=\sqrt{x_{1}\Lambda_{x}},
\end{equation}
what gives
\begin{equation}
S_{q}\left(Y_\pi\right)\geq\frac{q}{q-1}\ln\left(\frac{\Lambda_{y}}{\sqrt{x_{1}\Lambda_{x}}}\right),
\end{equation}
For $1\leq q\leq2$ we have $S_{q}\geq S_{2}$ what extends the above result providing the function $F_{\textrm{min}}=\min\left(\frac{q}{q-1};2\right)$.
\end{proof}

\section{The greatest singular value of $\Phi$}
\label{greatestsingular}
Before the proof of Theorem \ref{th:super-bound} we state the lemma.
\begin{lemma} \label{lemma:herm-enough}
For any matrix $M$ with $\|M\|_{\mathrm{HS}}^2 = \tr M M^{\dagger}  = 1$ there
exist a positive semi--definite matrix $P_M$ with $\|P_M\|_{\mathrm{HS}}^2=1$ such that
\begin{equation}
\|\Phi(M)\|_{\mathrm{HS}}^2 \leq \|\Phi(P_M)\|_{\mathrm{HS}}^2.
\end{equation}
\end{lemma}
\begin{proof}
First we will show, that one can choose hermitian matrix $H_M$, such that
\begin{equation}\label{eqn:hrem-bound}
\|\Phi(M)\|_{\mathrm{HS}}^2 \leq \|\Phi(H_M)\|_{\mathrm{HS}}^2.
\end{equation}
If we consider a decomposition of $M = H + i L$, where $H,L$ are hermitian matrices, we obtain
that
\begin{equation}
1 = \|M\|_{\mathrm{HS}}^2 = \|H\|_{\mathrm{HS}}^2 + \|L\|_{\mathrm{HS}}^2.
\end{equation}
Let ($0\leq p \leq 1$)
\begin{equation}
\|H\|_{\mathrm{HS}}^2 = p, \quad\ \|L\|_{\mathrm{HS}}^2 = 1-p,
\end{equation}
and define normalized hermitian matrices
\begin{equation}
H_0 = H/\sqrt{p}, \ \ L_0= L/\sqrt{1-p}.
\end{equation}
Now we write
\begin{equation}
\begin{split}
\|\Phi(M)\|_{\mathrm{HS}}^2 &= \|\Phi(H) + i \Phi(L)\|_{\mathrm{HS}}^2 \\
&= \|\Phi(H)\|_{\mathrm{HS}}^2 + \|\Phi(L)\|_{\mathrm{HS}}^2 \\
&= p \|\Phi(H_0)\|_{\mathrm{HS}}^2 + (1-p)\|\Phi(L_0)\|_{\mathrm{HS}}^2.
\end{split}
\end{equation}
Since $\|\Phi(M)\|_{\mathrm{HS}}^2$ is a convex combination  of
$\|\Phi(H_0)\|_{\mathrm{HS}}^2$ and $\|\Phi(L_0)\|_{\mathrm{HS}}^2$, therefore
$\|\Phi(H_0)\|_{\mathrm{HS}}^2 \geq \|\Phi(M)\|_{\mathrm{HS}}^2$ or
$\|\Phi(L_0)\|_{\mathrm{HS}}^2 \geq \|\Phi(M)\|_{\mathrm{HS}}^2$  and this
shows, that for any matrix $M$, with $\|M\|_{\mathrm{HS}}^2  = 1$ there exist a
hermitian matrix $H_M$, which satisfies~\eqref{eqn:hrem-bound}.

Now it is easy to notice, that by taking the absolute value of a hermitian
matrix $H_M$ we do not decrease the norm of the channel output, i.e. let
\begin{equation}
H_M = \sum \lambda_i \ketbra{\phi_i}{\phi_i}, \ \ |H_M|=\sum |\lambda_i| \ketbra{\phi_i}{\phi_i}.
\end{equation}
We have
\begin{equation}
\begin{split}
\|\Phi(H_M)\|_{\mathrm{HS}}^2
&= \sum \lambda_i \lambda_j \tr \Phi(\ketbra{\phi_i}{\phi_i})\Phi(\ketbra{\phi_j}{\phi_j}) \\
&\leq \sum |\lambda_i\lambda_j| \tr \Phi(\ketbra{\phi_i}{\phi_i})\Phi(\ketbra{\phi_j}{\phi_j})\\
&=\|\Phi(|H_M|)\|_{\mathrm{HS}}^2.
\end{split}
\end{equation}
\end{proof}

Now we are in position to prove  Theorem \ref{th:super-bound}.

\begin{proof}[Proof of Theorem \ref{th:super-bound}]
The definition of the greatest singular value of the super--operator reads
\begin{equation}
\sigma_1(\Phi) = \max_{\|M\|_{\mathrm{HS}} = 1} \|\Phi(M)\|_{\mathrm{HS}}.
\end{equation}
According to Lemma \ref{lemma:herm-enough} we can restrict the maximization to
positive semi--definite matrices $H$ such that $\|H\|_{\mathrm{HS}} = 1$.
Such matrix can be written as $H = \frac{\rho}{\sqrt{\tr \rho^2}}$ for $\rho \in \mathcal{M}_N$.
\begin{equation}
\begin{split}
\sigma_1(\Phi) &= \max_{\|M\|_{\mathrm{HS}} = 1} \|\Phi(M)\|_{\mathrm{HS}} \\
& = \max_{\rho \in \mathcal{M}_N } \left \|\Phi\left(\frac{\rho}{\sqrt{\tr \rho^2}}\right)\right\|_{\mathrm{HS}} \\
&= \max_{\rho \in \mathcal{M}_N} \sqrt{\frac{ \tr \Phi(\rho)^2}{\tr \rho^2}}.
\end{split}
\end{equation}
This proves the first part of Theorem \ref{th:super-bound}. In order to derive the second part we shall use the Kraus representation:
\begin{equation}
\Phi:\rho\mapsto\Phi\left(\rho\right)=\sum_{i}A_{i}\rho A_{i}^{\dagger},\qquad\sum_{i}A_{i}^{\dagger}A_{i}=\1,
\end{equation}
and write
\begin{equation}
\textrm{Tr}\Phi\left(\rho\right)^{2}=\sum_{i,j}\textrm{Tr}\left(A_{j}^{\dagger}A_{i}\rho A_{i}^{\dagger}A_{j}\rho\right),
\end{equation} where we also took an advantage from the trace invariance under cyclic permutations. Applying the matrix version of the Cauchy--Schwarz inequality $\textrm{Tr}XY\leq\sqrt{\textrm{Tr}X^{\dagger}X}\sqrt{\textrm{Tr}Y^{\dagger}Y}$
we obtain the bound
\begin{equation}
\textrm{Tr}\Phi\left(\rho\right)^{2}\leq\sum_{i,j}\sqrt{\textrm{Tr}\left(A_{j}^{\dagger}A_{i}\rho^{2}A_{i}^{\dagger}A_{j}\right)}\sqrt{\textrm{Tr}\left(A_{i}^{\dagger}A_{j}\rho^{2}A_{j}^{\dagger}A_{i}\right)}.\label{2}
\end{equation}
Obviously both families of matrices $A_{j}^{\dagger}A_{i}\rho^{2}A_{i}^{\dagger}A_{j}$
and $A_{i}^{\dagger}A_{j}\rho^{2}A_{j}^{\dagger}A_{i}$ are positive
semi--definite. We shall apply to (\ref{2}) the usual Cauchy--Schwarz
inequality to find that
\begin{equation}\label{sing1}
\textrm{Tr}\Phi\left(\rho\right)^{2}\leq\textrm{Tr}\left(\sum_{i,j}A_{i}^{\dagger}A_{j}\rho^{2}A_{j}^{\dagger}A_{i}\right).
\end{equation}
Next, we shall rearrange the right hand side of \ref{sing1} to the form
\begin{equation}
\textrm{Tr}\left(\sum_{i,j}A_{j}^{\dagger}A_{i}\rho^{2}A_{i}^{\dagger}A_{j}\right)=N\textrm{Tr\ensuremath{\rho^{2}}}\textrm{Tr}\left(\Phi\left(\frac{1}{N}\1\right)\Phi\left(\tilde{\rho}\right)\right),\label{9}
\end{equation}
where $\mathcal{M}_N\ni\tilde{\rho}=\frac{\rho^{2}}{\textrm{Tr\ensuremath{\rho^{2}}}}$. Finally, we use that expression to bound $\sigma_1$ given by the formula (\ref{singq1}):
\begin{equation}
\sigma_1(\Phi) \leq \sqrt{N}\max_{\tilde{\rho} \in \mathcal{M}_N} \sqrt{\textrm{Tr}\left(\Phi\left(\frac{1}{N}\1\right)\Phi\left(\tilde{\rho}\right)\right)}.
\end{equation}
The term $\textrm{Tr}\left(\Phi\left(\frac{1}{N}\1\right)
\Phi\left(\tilde{\rho}\right)\right)$ is bounded by the greatest eigenvalue of the
matrix $\Phi\left(\frac{1}{N}\1\right)$, which implies the desired result.



\end{proof}

\section{Estimating channel entropy}

\begin{figure}[ht]
\centering
\includegraphics{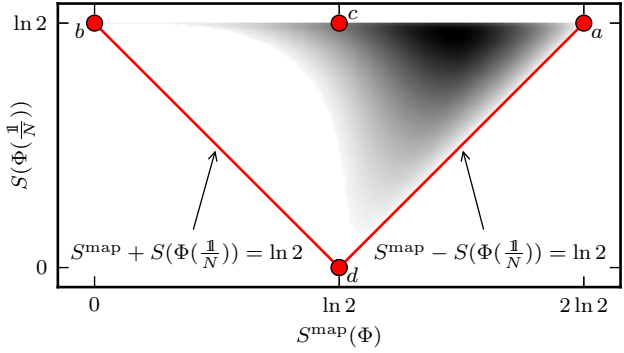}
\caption{(Color online) The region of allowed values of von Neumann entropies
$(S^{\map}(\Phi),S(\Phi(\frac{\idty}{N})))$ for one--qubit quantum channels
$\Phi$. Extremal lines denote the bounds proven in Proposition~\ref{prop102}.
The four distinguished points represent: $a$ -- completely depolarizing channel,
$b$ -- identity channel,  $c$ -- coarse graining channel and $d$ -- channels
describing spontaneous emission as shown in Fig~\ref{fig:sketch}.  The shading
denotes an estimation of the probability density of
$(S^{\map}(\Phi),S(\Phi(\frac{\idty}{N})))$, when $\Phi$ is chosen
randomly~\cite{frob}. The lower bound on the sum
$S^{\map}(\Phi)+S(\Phi(\frac{\idty}{N}))$ is in general not saturated. }
\label{estiment}
\end{figure}

The following inequality allows us to estimate the entropy of a channel. A similar estimation was recently formulated in \cite{zhang12}.

\begin{proposition}
\label{prop102}
For any quantum channel $\Phi$ acting on $\c M_N$ the following inequality holds
\begin{equation}
\ln{N}-S\left(\Phi\left(\frac{1}{N}\idty\right)\right)\leq S^{\map}(\Phi)
\leq \ln{N}+S\left(\Phi\left(\frac{1}{N}\idty\right)\right),
\label{upper}
\end{equation}
where $S$ denotes the von Neumann entropy.
Moreover, the right inequality is
satisfied for the R{\'e}nyi entropy
$S_q(\rho)=\frac{1}{1-q}\ln\tr{\rho^q}$ of an arbitrary order $q$.
\end{proposition}

\begin{proof}
For any quantum operation $\Phi$ the corresponding dynamical matrix $D_{\Phi}$
obeys the following relations \cite{BZ06}.
\begin{eqnarray}
 \tr_1 \frac1N D_{\Phi} &=&  \frac{1}{N}\idty , \\
  \tr_2 \frac1N D_{\Phi} &=& \Phi\left(\frac{1}{N}\idty \right).
\end{eqnarray}
Thus in the case of the von Neumann entropies the upper bound follows from
subadditivity, while the lower bound is a consequence of Araki--Lieb triangle
inequality \cite{Araki_Lieb}.
\end{proof}

In the case of the R{\'e}nyi entropies the upper bound follows directly from the weak
subadditivity~\cite{Dam02Renyi}. Although the lower bound (\ref{upper}) for the
von Neumann entropy of a map can not be directly extended for R{\'e}nyi entropies,
we provide another generalized bound, which holds for any $q\ge 0$,
\begin{equation}
\ln{N} - \ln \mathrm{rank}\left(\Phi \left(\frac{1}{N}\idty\right)\right)\leq
S^{\map}_q(\Phi).
\end{equation}
This lower bound for the generalized entropy of a map $S^{\map}_q$ follows also
from the weak subadditivity ~\cite{Dam02Renyi}.

Notice that in the special case of complete contraction
$\Phi_{\xi}:\rho\rightarrow\xi$ when the dynamical matrix has a form
$D_{\Phi_{\xi}}=\xi\otimes\idty$ the right hand side of inequality (\ref{upper})
is saturated. The bounds established by Proposition \ref{prop102} are
illustrated in Figure \ref{estiment}.


\end{document}